\documentclass[twocolumn,prl,showpacs,superscriptaddress,nofootinbib,notitlepage,amssymb,preprintnumbers]{revtex4-1}

\usepackage{amsmath,amssymb,amsthm}
\usepackage{braket}
\usepackage{empheq}
\usepackage{subcaption}

\usepackage{lipsum}

\usepackage[pdfstartview=FitH]{hyperref}
\hypersetup{
    colorlinks=true,       % false: boxed links; true: colored links
    linkcolor=cyan,          % color of internal links
    citecolor=magenta,        % color of links to bibliography
    filecolor=magenta,      % color of file links
    urlcolor=cyan,           % color of external links
    runcolor=cyan
}

\newcommand{\bes} {\begin{subequations}}
\newcommand{\ees} {\end{subequations}}
\newcommand{\ba}{\begin{eqnarray}}
\newcommand{\ea}{\end{eqnarray}}

\newcommand{\mrp}{\mathrm{p}}
\newcommand\norm[1]{\left\lVert#1\right\rVert}

\newcommand{\ketbra}[1]{|{#1}\rangle\langle#1|}

\newtheorem{theorem}{Theorem}
\newtheorem*{theorem*}{Theorem}

\newtheorem{lemma}{Lemma}[theorem]

\def\a{\alpha}

\begin{document}
\title{Robust universal Hamiltonian quantum computing using two-body interactions}
\author{Milad Marvian}

\date{\today}
%\affiliation{Department of Mechanical Engineering, Massachusetts Institute of Technology (MIT), Cambridge MA 02139, USA} 
\affiliation{Research Laboratory of Electronics, Massachusetts Institute of Technology, Cambridge, MA 02139, USA}

\author{Seth Lloyd}
\affiliation{Department of Mechanical Engineering, Massachusetts Institute of Technology, Cambridge MA 02139, USA} 
\affiliation{Research Laboratory of Electronics, Massachusetts Institute of Technology, Cambridge, MA 02139, USA}

\begin{abstract}
We present a new scheme to perform noise resilient universal adiabatic quantum computation using two-body interactions. 
To achieve this, we introduce a new family of error detecting subsystem codes whose gauge generators and a set of their logical operators --- capable of encoding universal Hamiltonian computations --- can be implemented  using two-body interactions. 
Logical operators of the code are used to encode any given computational Hamiltonian,  and the gauge operators are used to construct a penalty Hamiltonian
whose ground subspace is protected against local-errors. In contrast to previous approaches, the constructed penalty Hamiltonian does not necessarily commute with the encoded computational Hamiltonians, but for our construction the undesirable effect of the penalty Hamiltonian on the computation can be compensated by a simple modification of the implemented Hamiltonians.  

We also investigate whether a similar scheme can be constructed by encoding the computational Hamiltonian using only bare-logical operators of subsystem codes, to guarantee that computational Hamiltonian commutes with the penalty Hamiltonian. 
We prove a no-go theorem showing that restricting to two-body interactions and using any general CSS-type subsystem codes, such a construction cannot encode systems beyond an Ising chain in a transverse field. We point out that such a chain is universal for Hamiltonian-based quantum computation, but it is not universal for ground-state quantum computation.

\end{abstract}

\maketitle

\section{Introduction}
 
A general approach to reducing the effect of noise on any system is to introduce redundancy, in such a way that the larger system is more protected against the noise. 
Quantum codes capable of detecting the presence of a targeted set of errors provide a systematic way of encoding an arbitrary system into a larger system which is protected against errors.

A scheme to use quantum codes to protect a Hamiltonian-based quantum computation~\cite{aharonov_adiabatic_2007,albash2018adiabatic} was introduced by Jordan, Farhi, and Shor~\cite{jordan2006error}. In this scheme, a quantum code  that can detect the action of the system-bath interaction Hamiltonian on the system is chosen. Then, using the stabilizers of the code a penalty Hamiltonian whose ground subspace is protected against the induced errors is constructed. By construction, logical operators of any stabilizer quantum code commute with the stabilizer operators and hence with the penalty Hamiltonian. Therefore the structure of the code allows using logical operators of the code to perform arbitrary computation in the ground subspace of the penalty Hamiltonian, which penalizes any excitation out of the codespace. This approach has been generalized and its performance in non-Markovian \cite{Bookatz:2014uq} and Markovian \cite{marvian2017error,lidar2019arbitrary} environments have been studied.

However, the experimental feasibility of any such scheme is heavily dependent on the required extra resources. In addition to the extra qubits, an important consideration is the type of interactions, such as locality of the interactions that need to be implemented. Even if the dominant errors on the system are one-local, a penalty Hamiltonian that is constructed using any (subspace) stabilizer codes requires at least four-local interactions~\cite{jordan2006error}. More generally, any commuting Hamiltonian with a ground subspace protected against one-local errors cannot be implemented using two-body interactions~\cite{Marvian:2014nr}. As strong controllable high-weight interactions are hard to engineer, it is important to know whether error suppression is possible using lower-weight interactions. 

Going beyond commuting Hamiltonians, recently it has been shown that using subsystem codes one can construct a non-commuting two-body Hamiltonian capable of suppressing general one-local errors~\cite{Jiang:2015kx,Marvian-Lidar:16}. This reduction in the locality of the penalty Hamiltonian makes them experimentally more feasible. For example, using Josephson phase-slip qubits implementation of such a penalty Hamiltonians, with strong static two-body interactions and no local fields,  can increase the effective coherence time of qubits considerably  \cite{kerman2019superconducting}.
Keeping the penalty Hamiltonian two-local, we can ask what is the minimum locality of an encoded computational Hamiltonian. Any logical operator has to be at least two-local to be different from one-local errors. Therefore, encoding a Hamiltonian by individually replacing each qubit with its logical counterpart (block encoding), would encode a two-local interaction into a four-local interaction. 
Two-local constructions to embed Hamiltonians while suppressing general one-local errors have been found for limited Hamiltonians, such as an Ising chain in a transverse field~\cite{Marvian-Lidar:16}.

The ground-state computation using the Ising chain in a transverse field is unlikely to be universal, but with a fast control of local fields, and therefore going out of the ground space, the chain is already powerful enough to simulate any other Hamiltonian and hence is computationally universal \cite{Benjamin:01,lloyd2018quantum}. 
Until now, it has been left open whether a non-perturbative two-local construction that can encode a universal adiabatic quantum computation while suppressing  one-local errors is possible.
Here we answer this question affirmatively. 

Here we present a subsystem code whose gauge operators, single qubit logical operators, and products of two logical operators of the same type (up to a gauge operator) are all two-local operators. Then we show how this code can be used to encode universal adiabatic quantum computations while suppressing one-local errors. Unlike previous constructions, the penalty Hamiltonians and the encoded Hamiltonians do not necessarily  commute, but we will show that for our construction the undesirable effect of the penalty Hamiltonian on the computation can be compensated by a simple modification of the implemented Hamiltonians. 
In addition, in a no-go theorem, we  show  that requiring the penalty Hamiltonian to commute with the encoded Hamiltonian limits the geometry of the Hamiltonians that can be protected. In particular, we show that, using general all-two-local CSS-type subsystem codes, it is not possible to protect Hamiltonians beyond an Ising chain in a transverse field.

\textit{Introduction to subsystem codes}-- We first briefly review subsystem stabilizer codes~\cite{poulin_stabilizer_2005}  which are the generalization of subspace stabilizer codes~\cite{Gottesman:1996fk}. In short, subsystem stabilizer codes can be thought as subspace stabilizer codes where some of the logical qubits are not used to store information, ignoring the errors that occur on these logical qubits.  

More formally, any subspace stabilizer code is defined by an Abelian group $\mathcal{S}$ of Pauli operators, where the codespace is the simultaneous +1 eigenstates of the group elements. The set of operators that are included in $\mathcal{S}$ determine the error-detecting properties of the code.
Any Pauli error that anti-commutes with at least one element of $\mathcal{S}$ takes a codestate to a state out of the codespace and hence is a detectable error.
The set of Pauli operators that are not in $\mathcal{S}$ but keep it invariant, $\mathcal{C}(\mathcal{S})/\mathcal{S}$,  perform logical operations in the codespace. 
In subsystem codes, these operators are partitioned into two commuting set of Pauli operators:  set of gauge operators $\mathcal{A}$ and set of logical operators $\mathcal{L}$. This partitioning induces a subsystem structure in the codespace~\cite{ZLL:04}.

The (non-Abelian) gauge group is defined as $\mathcal{G}=\braket{\mathcal{S},\mathcal{A}}$. Logical operators that 
preserve the codespace and act trivially on the gauge subsystem are called bare-logical operators, as opposed to dressed-logical operators that also preserve the codespace
but can act non-trivially in the gauge subsystem. 

The relaxed criterion, allowing errors to act on the gauge subsystem in the codespace, makes subsystem codes more powerful in some cases. Subsystem codes such as the Bacon-Shor code \cite{bacon2006operator, shor1995scheme} require simpler syndrome measurements, leading to surprisingly good error correction performances \cite{aliferis2007subsystem,Napp:2013,li2018direct}.
Furthermore, it has been shown that under the locality constraint, subsystem codes can encode more information comparing to subspace codes~\cite{PhysRevA.83.012320}.

To protect a given computational Hamiltonian $H_s$ against a set of errors, we choose a subsystem code that can detect the errors. Then a (non-commuting) penalty Hamiltonian is constructed using the elements of the gauge group of the code:
\ba
H_p=- \sum_{g_i \in \mathcal{G}} g_i \, .
\ea
By construction, the ground subspace of $H_p$ is protected against the targeted set of errors. To perform the desired computation, an encoded computational Hamiltonian $\bar{H}_s$ is constructed by replacing each operator in $H_s$ with the corresponding logical operator of the code. $\bar{H}_s$ performs the desired computation in the ground subspace of the penalty Hamiltonians, which is a protected subspace against errors. 

Encoding $\bar{H}_s$ using bare-logical operators of a code guarantees that it commutes with  the penalty Hamiltonian $H_p$, and hence the penalty Hamiltonian does not interfere with the computation. 
Constructing $\bar{H}_s+E_p H_p$ using subsystem codes, in Ref.~\cite{Marvian-Lidar:16} general conditions guaranteeing protection against errors  and also  performance bounds have been derived.
In addition, some examples of two-local Hamiltonians capable of suppressing one-local errors were introduced. One of the examples is Ising chain (without local field) in a transverse field. It was left open whether it is possible to extend such a two-local construction to encode universal Hamiltonian computation while suppressing local errors.  

We first show that using bare-logical operators of any CSS subsystem code whose 
gauge group can be generated using two-local interactions, including the generalized Bacon-Shor subsystem codes~\cite{PhysRevA.83.012320,yoder2019optimal}, one cannot encode systems beyond the transverse field Ising chain introduced in Ref.~\cite{Marvian-Lidar:16}.

\begin{theorem}
Consider any nontrivial (with distance at least two) stabilizer subsystem code of CSS type with a gauge group that can be generated using  $XX$ and $ZZ$ interactions. Then the weight of an X-type (Z-type)  single-qubit bare-logical operator is lower-bounded by the number of Z-type (X-type) bare-logical operators acting on its supporting logical qubits.
\end{theorem}
\begin{proof}
See the Supplementary Material (SM).
\end{proof}
For example, consider a subsystem code which has a set of weight-two gauge group generators, and has weight-two  bare-logical operators $\{\bar{Z}_1, \bar{Z}_1\bar{Z}_2,\bar{Z}_1\bar{Z}_3\}$. Then $\bar{X}_1$ has to be at least three-local.
Therefore even to generate bare-logical operations equivalent to a simple set of interactions such as $\{\bar{Z}_1, \bar{Z}_1\bar{Z}_2,\bar{Z}_1\bar{Z}_3, \bar{X}_1\}$ or $\{\bar{Z}_1\bar{Z}_2, \bar{Z}_1\bar{Z}_3,\bar{Z}_1\bar{Z}_4, \bar{X}_1\}$,  physical interactions with a weight larger than two are required. 
From this we conclude that using this scheme with bare-logical encoding, one cannot go beyond the encoding of an Ising \textit{chain} in a transverse field as constructed in Ref.~\cite{Marvian-Lidar:16}.

%As a consequence, using generalized Bacon-Shor subsystem code and only using two-body interactions (with bare-logical encoding) one cannot generate logical operations equivalent to even a simple set of interactions such as $\{Z_1, Z_1Z_2, Z_1Z_3, X_1\}$ or $\{Z_1Z_2, Z_1Z_3, Z_1 Z_4, X_1\}$ and hence cannot go beyond the encoding of an Ising \textit{chain} in a transverse field as constructed in Ref.~\cite{Marvian-Lidar:16}.

Ising chain in a transverse field, with fast control on the transverse field, 
is already powerful enough to simulate any other Hamiltonian and hence is computationally universal \cite{Benjamin:01,lloyd2018quantum}. But ground-state computation using the Ising chain in a transverse field is unlikely to be universal \cite{Cubitt:2016vl}.\footnote{Since, for example, any Ising chain with transverse field can be efficiently converted to a stoquastic Hamiltonian.} 

To circumvent this limitation, we consider encoding using dressed logical operators. 
 
\textit{A fully two-local code}-- 
Here we present a $[[6k,2k,2]]$ quantum subsystem code (it encodes $2k$ logical qubits into $6k$ physical qubits and can detect arbitrary 1-local errors) with the property that a complete set of its gauge group generators, single qubit logical operators $\{\bar{X}_i,\bar{Z}_i\}$, and also the product of any two logical operators of the same type $\{\bar{X}_i \bar{X}_j,\bar{Z}_i \bar{Z}_j,\bar{Y}_i \bar{Y}_j \}$ (up to a gauge operator) can all be implemented using two-body interactions. In Figure~\ref{Fig_code}, the defining properties of this code are presented. In particular,  the gauge group of the code is  the group generated by
\begin{eqnarray} \label{gauge}
\mathcal{G}=\braket{X_{B_i} X_{R_i}, X_{L_i} X_{L_{i+1}}, Z_{B_i} Z_{L_i}, Z_{R_i} Z_{R_{i+1}}},\nonumber
\end{eqnarray}
where $1\leq i \leq 2k$.
%\begin{eqnarray} \label{gauge}
%X_{B_i} X_{R_i},\\
%X_{L_i} X_{L_{i+1}}, \nonumber \\
%Z_{B_i} Z_{L_i},\nonumber \\
%Z_{R_i} Z_{R_{i+1}}\nonumber
%\end{eqnarray}
Single-qubit logical operators can be defined as $\bar{X}_i=X_{B_i} X_{L_i}$ and $\bar{Z}_i=Z_{B_i} Z_{R_i}$. Product of  two logical operators  $\bar{X}_i$ and  $\bar{X}_j$ becomes $X_{B_i} X_{L_i}X_{B_j} X_{L_j}$. Since $X_{L_i}X_{L_j}$ is an element of the gauge group,  the two-body interaction $X_{B_i} X_{B_j}$ implements  $\bar{X}_i  \bar{X}_j$ up-to a gauge operator. Similarly $\bar{Z}_i\bar{Z}_j$ and $\bar{Y}_i \bar{Y}_j$ can be implemented using two-body interactions, up-to a gauge operator. 
This code can be understood as an example of Bravyi's generalization of Bacon-Shor code ~\cite{PhysRevA.83.012320} (see the SM for details), also it can be interpreted as a subsystem version of the code introduced in Ref.~\cite{ganti2014family}. 
\begin{figure*}[!t]
\begin{center}
\includegraphics[height=0.30\textheight]{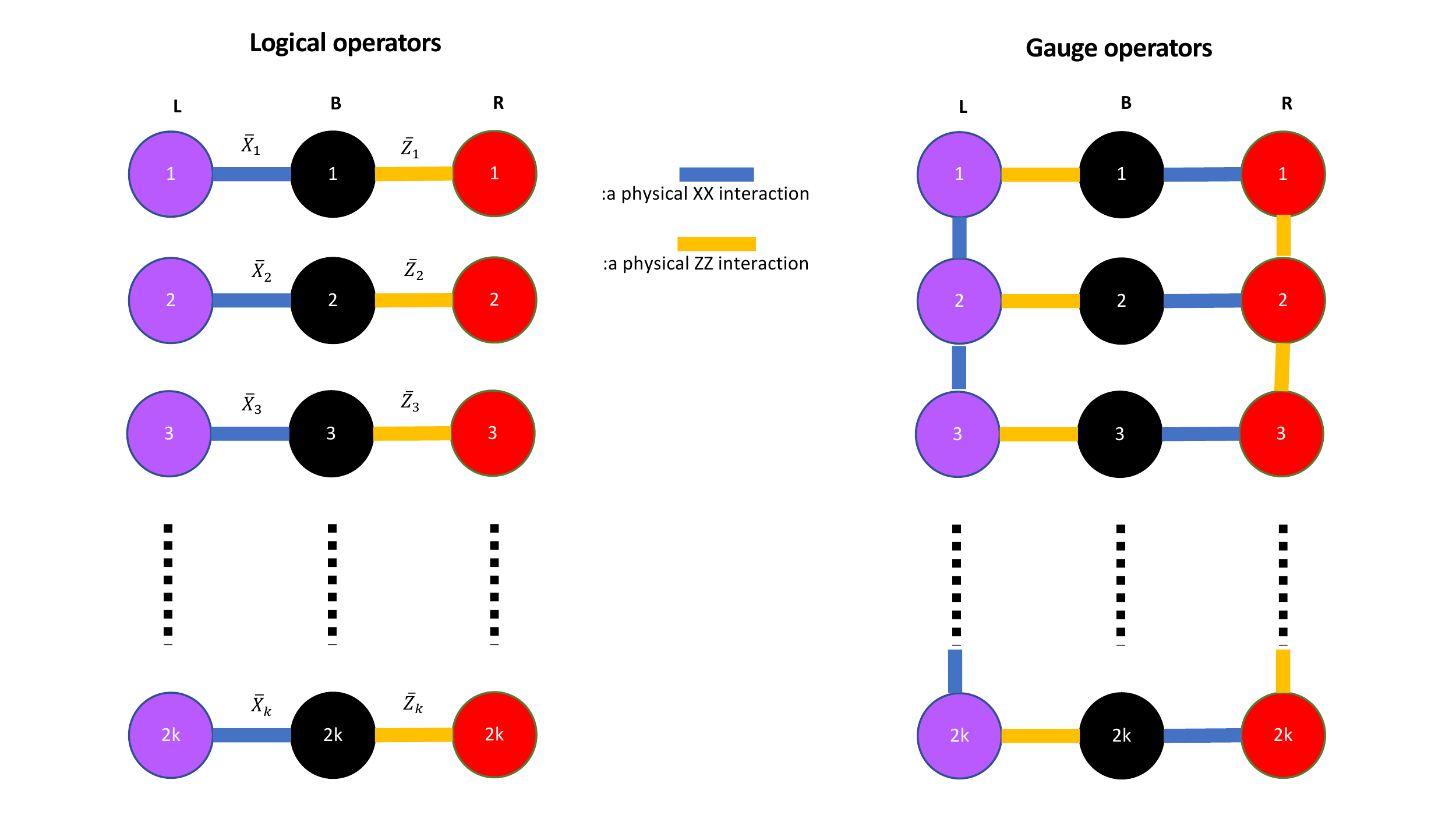}
\caption{\textbf{Description of the $[[6k,2k,2]]$ subsystem code:} Blue lines represent $XX$ interactions and orange lines represent $ZZ$ interactions. A set of generators of the gauge group of the code are presented in right panel. A set of logical operators are presented in the left  panel.}
\label{Fig_code}
\end{center}
\end{figure*}

%\section{Error suppression}
%In this section, we show that two-body physical interactions are sufficient to perform noise resilient universal quantum computation. A Hamiltonian $H_s(t)$ that consists of \{X, Z, XX, ZZ\}-interactions are capable of generating universal quantum computation \cite{Biamonte:07}, therefore in what follows, we show how to implement and protect such a Hamiltonian. We accomplish this by introducing an error detecting subsystem code whose gauge group generators and logical \{$X$,$Z$,$XX$,$ZZ$\}-interactions are weight-two operators. 
%
%$YY$ \cite{Lloyd:2016}

\textit{Encoding}-- 
Suppose we want to encode the following universal Hamiltonian:
\begin{equation} \label{universalH}
H_s(t)= \sum_i a_i X_i +  \sum_i b_i Z_i + \sum_{ij} c_{ij} X_i X_j 
+  \sum_{ij} d_{ij} Z_i Z_j.
\end{equation}

If we replace each operator in this Hamiltonian with the corresponding  bare-logical operators of the introduced code, the resulting Hamiltonian, 
\begin{equation} \label{barHs}
\bar{H}_s(t)= \sum_i a_i \bar{X}_i +  \sum_i b_i \bar{Z}_i + \sum_{ij} c_{ij} \bar{X}_i \bar{X}_j 
+  \sum_{ij} d_{ij} \bar{Z}_i \bar{Z}_j,
\end{equation}
would be at least four-local. (Implementing  $\bar{X}_i\bar{X}_j$ and $\bar{Z}_i \bar{Z}_j$  require at least four-body interactions.) 
%Rather than using the bare-logical operator of the code, we can use dressed logical operators. 
Instead, we implement the following two-body physical Hamiltonian:
\ba
\hat{H}_s(t)&=& \sum_i a_i X_{L_i}X_{B_i}  +  \sum_i b_i Z_{B_i} Z_{R_i}  \nonumber \\
&+& \sum_{ij} c'_{ij} X_{B_i} X_{B_j}   +  \sum_{ij} d'_{ij} Z_{B_i} Z_{B_j},
\ea
where the notation of Fig.~\ref{Fig_code} is used. Later, we discuss how to choose $c'_{ij}$ and $d'_{ij}$. Re-writing this Hamiltonian using the logical and gauge operators, we have
\ba
\hat{H}_s(t)&=& \sum_i a_i \bar{X}_i +  \sum_i b_i \bar{Z}_i  \nonumber \\
&+& \sum_{ij} c'_{ij} \bar{X}_i \bar{X}_j g^x_{ij}  +  \sum_{ij} d'_{ij} \bar{Z}_i \bar{Z}_j g^z_{ij}.
\ea
In other words, this physical two-local Hamiltonian implements our desired computational Hamiltonian (up to some gauge operators) using dressed-logical operators of the code. Using dressed-logical operators, as opposed to bare-logical operators, we can avoid the no-go theorem.

In general, the dressed-logical operators $\bar{Z}_i \bar{Z}_j g^z_{ij}$ and $\bar{X}_i \bar{X}_j g^x_{ij}$ can couple the gauge subsystem and the information subsystem. Since $[\hat{H}_s(t) , H_p] \neq 0$, the penalty Hamiltonian can interfere with our desired computation. But surprisingly, in what follows we show that in the large penalty limit this coupling effect can be easily compensated with a simple rescaling of the coefficients $c_{ij}$ and $d_{ij}$. 

Although in Figure \ref{Fig_code} the qubits appear to be on a chain,  the encoding described here works for  Hamiltonians on a 2D lattice (see the SM).
It is important to note that this scheme encodes a geometrically local Hamiltonian into geometrically local Hamiltonian. 

Also note that all the constructions and derivations presented here can be extended to the case where $H_s$ in Eq.~\ref{universalH},  the Hamiltonian we want to protect,  includes $Y_i Y_j$ interactions as well. This is a result of the fact that a logical $\bar{Y}_i \bar{Y}_j$ interaction can be implemented using $Y_{B_i} Y_{B_j}$ physical interactions. 

\textit{ Large penalty limit}--
In this section, we prove that in the large energy penalty limit, not only the system becomes decoupled from the bath,  but it  also faithfully performs the computation we encode. More precisely,  in the Supplementary Material --without using perturbation theory-- we prove that by increasing the energy penalty $E_p$ the evolution generated by 
$\hat{H}_s(t) + E_p H_p +H_B + H_{SB}$ becomes arbitrary close to the decoupled evolution generated by $\bar{H}_s(t) +H_B+E_p H_p $.  
But here it is instructive to study the effective Hamiltonian in the limit of large $E_p$, using first order perturbation theory. Denoting the ground subspace of the penalty Hamiltonians $H_p$ by $\Pi_0$, and assuming that the initial state is prepared in this subspace, $\hat{H}_s(t) + E_p H_p +H_B + H_{SB}$  in the large $E_p$ limit effectively becomes $\Pi_0 (\hat{H}_s(t) +H_B + H_{SB})\Pi_0$. 

Since $H_p$ acts non-trivially only on the system, clearly the bath Hamiltonian $H_B$ commutes with $\Pi_0$. In what follows we will show that $\Pi_0 H_{SB}\Pi_0=0$ and $\Pi_0 \hat{H}_s(t) \Pi_0=\Pi_0 \bar{H}_s(t) \Pi_0=\Pi_0 \bar{H}_s(t)$. Combing these confirms that the effective Hamiltonian becomes $\Pi_0(\bar{H}_s(t) +H_B) $ which is the desired computation decoupled from the environment. 

The rest of this section is devoted to prove that $\Pi_0 H_{SB}\Pi_0=0$ and $\Pi_0 \hat{H}_s(t) \Pi_0=\Pi_0 \bar{H}_s(t)$. These two properties are consequences of the specific form of the penalty Hamiltonian, i.e., the fact that the penalty Hamiltonian does not have positive off-diagonal elements and also includes a complete set of gauge group generators.

\begin{lemma} \label{lemma1gaugeHamiltonian}
Let $H_p=-\sum_{g_i \in G'} g_i$, where $G'$ is a set of X-type and Z-type gauge operators that can generate the gauge group $G$. Denote the projector to  ground subspace of $H_p$ by $\Pi_0$. Then
\begin{enumerate}
  \item Any ground state of $H_p$ is stabilized by the stabilizers of the subsystem code, defined as  $S=G \cap Centralizer(G)$ \cite{burton2018spectra,Bravyipriv}.

%  \item The ground state of $H_p$ is exactly $2^k$ degenerate. In other words, the ground state of $H_p$ in the gauge subsystem is unique (see Eq.\ref{Eq:H}). 
  \item $\Pi_0 g_i  \Pi_0$ is proportional to $\Pi_0$, i.e. $\Pi_0 g_i  \Pi_0=\a_i \Pi_0$.
\end{enumerate}
\end{lemma}
\begin{proof}
See the Supplementary Material.
\end{proof}

The first part of the Lemma guarantees that the ground subspace of the penalty Hamiltonian is in the codespace and hence can detect any local error $\sigma_\a$, i.e. $\Pi_0 \sigma_\a \Pi_0=0$. By assumption, the system part of the interaction Hamiltonian is one-local, i.e. $H_{SB}=\sum \sigma_\a \otimes B_\a$, and therefore 
\ba
\Pi_0 H_{SB} \Pi_0=0.
\ea 

To simplify $\Pi_0 \hat{H}_s(t) \Pi_0$, we note that  by definition, bare-logical operators leave the ground subspace of the penalty Hamiltonian invariant, i.e., we have $[\bar{X}_i,\Pi_0]=[\bar{Z}_i,\Pi_0]=0$, and therefore 
\ba
\Pi_0 \hat{H}_s(t) \Pi_0
= \Pi_0[\sum_i a_i \bar{X}_i &+&  \sum_i b_i \bar{Z}_i  \nonumber \\
+ \sum_{ij} c'_{ij} \bar{X}_i \bar{X}_j \Pi_0 g^x_{ij} \Pi_0 &+&  \sum_{ij} d'_{ij} \bar{Z}_i \bar{Z}_j \Pi_0 g^z_{ij} \Pi_0].
\ea

The second part of Lemma~\ref{lemma1gaugeHamiltonian} states that the effect of gauge operators in ground subspace of the penalty Hamiltonian is just a trivial energy shift:  $\Pi_0 g_{ij}\Pi_0=\alpha_{ij} \Pi_0$, where $\alpha_{ij}$ is a scalar. Therefore, we have 
\ba
\Pi_0 \hat{H}_s(t) \Pi_0
&=& \Pi_0[\sum_i a_i \bar{X}_i +  \sum_i b_i \bar{Z}_i  \nonumber \\
&+& \sum_{ij} \alpha^x_{ij}  c'_{ij} \bar{X}_i \bar{X}_j  +  \sum_{ij} \alpha^z_{ij}  d'_{ij} \bar{Z}_i \bar{Z}_j ].
\ea
Choosing $c'_{ij} = \frac{c_{ij}}{\alpha^x_{ij} }$ and $d'_{ij} = \frac{d_{ij}}{\alpha^z_{ij} }$ we get $\Pi_0 \hat{H}_s(t) \Pi_0=\Pi_0 \bar{H}_s(t)\Pi_0$, i.e., the effective Hamiltonian is identical to the desired Hamiltonian of Equation \ref{barHs}.

It is crucial that the rescaling ratios $\a_{ij}$ are independent of the computation; they only depend on the penalty Hamiltonian. 
For any penalty Hamiltonian these rescaling ratios can be measured (and stored) as part of a one-time calibration process. Also these coefficient are expected to be non-vanishing (see the SM), for example for the symmetric compass model, a quantum model closely related to the Bacon-Shor code \cite{li20192d}, this ratio is around 0.5~\cite{brzezicki2013symmetry}.

Note that the penalty Hamiltonian can be gapless, i.e., energy gap protection against errors decreases with increasing the system size. Therefore to maintain the same level of protection, we have to scale up the interactions of the penalty Hamiltonian. For the transverse field Ising spin Chain, the gap decreases linearly with the number of qubits~\cite{Marvian-Lidar:16} and for the full compass model it exhibits a power-law behavior~\cite{brzezicki2013symmetry}. Nevertheless, in general, the minimum gap of the computational Hamiltonian can decrease exponentially. (For example, for an adiabatic quantum optimizer with a final Hamiltonian that encodes an NP-hard problem~\cite{farhi_quantum_2000}.) Therefore the polynomially decreasing gap of the penalty Hamiltonian can provide significant protection against error for a large range of system sizes.
 
 \textit{Discussions.}-- 
Adiabatic quantum computation offers some intrinsic robustness against noise due to the energy gap in the system~\cite{childs_robustness_2001,aharonov_adiabatic_2007,Lloyd:2008zr}. %Although  the relevant gap against thermal errors decreases only linearly in the length of computation~\cite{Lloyd:2008zr}, 
The Hamiltonians constructed to solve NP-hard problems, such as the quantum adiabatic optimization algorithms~\cite{farhi_quantum_2000}, have exponentially decreasing gaps and therefore additional methods of error suppression are desirable. 
In this work, we have shown how to embed and protect universal Hamiltonians against local errors,  using  two-body interactions. For this purpose, we introduced a new error detecting code whose logical operators capable of universal AQC have a minimum weight. 
Since stable two-body interactions are experimentally more accessible than higher-local interactions, we expect an immediate application of our scheme in near-term quantum devices.
Furthermore, our approach shows how to use simpler two-body interactions to effectively generate protected logical interactions with higher locality, and therefore we expect our approach to be useful in developing noise-resilient perturbative gadgets. 

%More native to the device in hand.

\textit{Note added:}
A result similar to Theorem 1, but limited to Bravyi's generalization of Bacon-Shor code, has been independently derived by P. Lisonek, A. Roy, and S. Trandafir using a different approach.

\textit{Acknowledgments.}-- We thank Daniel Lidar and Andrew Kerman for their comments on earlier versions of this manuscript. We thank Sergey Bravyi and Theodore J. Yoder for useful discussions. 
The research is based upon work (partially) supported by the Office of the Director of National Intelligence (ODNI), Intelligence Advanced Re- search Projects Activity (IARPA), via the U.S. Army Research Office contract W911NF-17-C-0050. The views and conclusions contained herein are those of the authors and should not be interpreted as necessarily representing the official policies or endorsements, either expressed or implied, of the ODNI, IARPA, or the U.S. Government. The U.S. Government is authorized to reproduce and distribute reprints for Governmental purposes notwithstanding any copyright annotation thereon.

\bibliographystyle{IEEEtran}%apsrev
\bibliography{refs}

\onecolumngrid
\appendix
% \fontsize{12}{18}\selectfont
\section{Error bound}
We prove that by increasing the energy penalty $E_p$ the evolution generated by 
$\hat{H}_s(t) + E_p H_p +H_B + H_{SB}$ becomes arbitrary close to the decoupled evolution generated by $\bar{H}_s(t) +H_B+E_p H_p $. To do so, we invoke the following theorem of Ref.~\cite{Marvian-Lidar:16}:

\begin{theorem} (Theorem 2 in Ref.~\cite{Marvian-Lidar:16})
Let $H_p=\sum_i \lambda_i \Pi_i$ be the eigen-decomposition of $H_p$, and the projector to the initial state be $P=\sum_{i \in \mathcal{I}}\Pi_i$. Denote the unitary evolutions generated by $H_V= H_0+E_{\mrp} H_{\mrp} + V$ and $H_W = H_0+E_{\mrp} H_{\mrp} + W$, where $W=\sum_{i\in \mathcal{I}} \Pi_i V \Pi_i$, respectively by $U_V$ and $U_W$. Assume $[H_0,H_p]=[H_0,P]=0$ and denote the final time by $T$. Then:
\ba
\norm{U_V(T) P- U_W(T) P} \leq \\
\norm{K(T)} + T(\norm{V}+\norm{W})&&\sup_{t}{\norm{K(t)}}+T\sup_{t}\norm{[K(t),H_0(t)]} \ ,
\label{eq:UV-UW-again}
\ea
where
\ba
\norm{K(t)} \leq  \frac{2}{E_{\mrp}} \sum_{a\neq a'} \frac{\norm{V-W}}{|\lambda_{a}-\lambda_{a'}|} . 
\ea
\end{theorem}

%A simple bound can be derived by plugging in $H_0=0$, $P=\Pi_0$, and $V=\hat{H}_s(t)+H_B+H_{SB}$ and hence $W=\Pi_0V\Pi_0$. Since $||V||$ and $||W||$ are independent of the $E_p$, in the large penalty limit the evolution generated by $E_{\mrp} H_{\mrp} + V=\hat{H}_s(t)+E_{\mrp} H_{\mrp}+H_B+H_{SB}$ gets close to $E_{\mrp} H_{\mrp} + W=E_{\mrp} H_{\mrp}+ \Pi_0 (\hat{H}_s(t)+H_B +H_{SB})\Pi_0$. 
%
%As it is discussed in the main text, using the error detecting property of the code, we have $\Pi_0 H_{SB} \Pi_0=0$. Also $\Pi_0 \hat{H}_s(t)\Pi_0=\Pi_0 \bar{H}_s(t)\Pi_0$ as desired.
%
%A tighter bound can be achieved by treating the Hamiltonian that commutes with the penalty Hamiltonian separately. In this case, again using the notation of Theorem 2 in Ref.~\cite{Marvian-Lidar:16}, we can choose $H_0=H_B+\sum_a \Pi_a \hat{H}_s(t) \Pi_a$, $V=\sum_{a\neq a'} \Pi_a \hat{H}_s(t) \Pi_{a'}+H_{SB}$,   $W=\Pi_0V\Pi_0=0$, and $P=\Pi_0$. 
%
%\newpage

Our desired bound can be derived by plugging in $H_0=\sum_a \Pi_a \hat{H}_s(t) \Pi_a+H_B$, $P=\Pi_0$, $V=\sum_{a\neq a'} \Pi_a \hat{H}_s(t) \Pi_{a'}+H_{SB}$, and  $W=\Pi_0 V \Pi_0=\Pi_0 H_{SB} \Pi_0$.

Since $||V||$ and $||W||$ are independent of  $E_p$, in the large penalty limit the evolution generated by $H_0+E_{\mrp} H_{\mrp} + V=\hat{H}_s(t)+E_{\mrp} H_{\mrp}+H_B+H_{SB}$ gets close to $H_0+E_{\mrp} H_{\mrp} + W=\sum_a \Pi_a \hat{H}_s(t) \Pi_a+E_{\mrp} H_{\mrp} +H_B+ \Pi_0 H_{SB}\Pi_0$.
As discussed in the main text, using the error detecting property of the code, we have $\Pi_0 H_{SB} \Pi_0=0$. 

Since the remaining terms all commute with the projector to the initial state, we have $\mathcal{T}exp\big(\int^T_0 (\sum_a \Pi_a \hat{H}_s(t) \Pi_a + E_p H_p +H_B   ) dt \big) \Pi_0= Texp\big(\int^T_0 (\Pi_0 \hat{H}_s(t) \Pi_0 + E_p H_p +H_B ) dt \big) \Pi_0 $, and hence the effective Hamiltonian is equal to $\Pi_0 \hat{H}_s(t) \Pi_0+E_{\mrp} H_{\mrp} +H_B$. As discussed in the main text, we have $\Pi_0 \hat{H}_s(t) \Pi_0=\Pi_0 \bar{H}_s(t) \Pi_0$ which concludes the proof.

%\newpage
%
%\ba
%||Texp\big(\int^T_0 (\hat{H}_s(t) + E_p H_p +H_B + H_{SB}  ) dt \big) \Pi_0- Texp\big(\int^T_0 (H_B+\sum_a \Pi_a \hat{H}_s(t) \Pi_a + E_p H_p ) dt \big) \Pi_0|| \\
%\leq  \frac{2}{E_p} ( T(||V||+||W||)+1) \sum_{a\neq a'} \frac{||V-W||}{\lambda_a-\lambda_{a'}}
%\ea
%
%Using the notation of Theorem 2 of Ref.~\cite{Marvian-Lidar:16}, we choose $V=\hat{H}_s(t)+H_B+H_{SB}$ and $P=\Pi_0$ and $H_0=0$. 
%Then defining $W=\Pi_0V\Pi_0$, we have: 
%\ba
%||Texp\big(\int^T_0 (V + E_p H_p ) dt \big) \Pi_0- Texp\big(\int^T_0 (W + E_p H_p ) dt \big) \Pi_0|| \\
%\leq  \frac{2}{E_p} ( T(||V||+||W||)+1) \sum_{a\neq a'} \frac{||V-W||}{\lambda_a-\lambda_{a'}}
%\ea

\section{Spectral properties of the penalty Hamiltonian}
\begin{lemma} \label{lemma1gaugeHamiltonian_appendix}
Let $H_p=-\sum_{g_i \in G'} g_i$, where $G'$ is a set of X-type and Z-type gauge operators that can generate the gauge group $G$. Denote the projector to  ground subspace of $H_p$ by $\Pi_0$. Then
\begin{enumerate}
  \item Any ground state of $H_p$ is stabilized by the stabilizers of the subsystem code, defined as  $S=G \cap Centralizer(G)$.

  \item The ground state of $H_p$ is exactly $2^k$ degenerate, where $k$ is the number of logical qubits of the code. %(In other words, the ground state of $H_p$ in the gauge subsystem, defined in Eq.\ref{Eq:H},  is unique. )
  
  \item $\Pi_0 g_i  \Pi_0$ is proportional to $\Pi_0$, i.e. $\Pi_0 g_i  \Pi_0=c_i \Pi_0$.
\end{enumerate}
\end{lemma}
The observation that the ground space of the penalty Hamiltonian is  in the codespace has been pointed out in Refs.~\cite{burton2018spectra,Bravyipriv}. The second and third statements are appearing for the first time but are related to Lemma 5 and Proposition 7 of Ref~\cite{burton2018spectra}.

\begin{proof}
First note that all the bare-logical operators and stabilizers commute with all $g_i$ and hence with $H_p$. We denote the eigenvalues of the Z-type bare-logical operator $\bar{Z}_i$ by $z_i \in \{-1,1\}$, and  denote the eigenvalues of  the Z-type stabilizers $S^z_i$ by $s_i \in \{-1,1\}$. We collectively  represent the values of $z_i$ by $\bold{z}$, and the values of $s_i$ by $\bold{s}$.
Therefore we have
\ba \label{eq:Hpdirectsum}
H_p=\bigoplus_{\bold{z}} \bigoplus_{\bold{s}}   H_{\bold{z},\bold{s}} \, .
\ea

Note that $H_p$ is independent of the Z-type bare-logical operator. Therefore if we use the sectors indexed by $\bold{z}$ to represent $H_p$, it becomes a block-diagonal matrix with identical blocks of 
\ba \label{Eq:H}
H=\bigoplus_{\bold{s}}   H_{\bold{z},\bold{s}}
\ea 
(hence  omitting the $\bold{z}$ subscript from $H$). To prove the second and third statements of the lemma,  we  prove that the ground state of $H$ is non-degenerate.

In every $(\bold{z},\bold{s})$-subspace the restricted Hamiltonian $H_{\bold{z},\bold{s}}$ is an irreducible symmetric matrix with non-positive off-diagonal elements in the computational basis.
Non-positivity of the off-diagonal elements follows from the fact the only terms in $H_p$ that generate off-diagonal elements are X-type operators with negative coefficients. 
Since a complete set of X-type generators are  included in  $H_p$,  in each $(\bold{z},\bold{s})$-sector we can transverse from any computational state  to any other computational state using a sequence of  X-type gauge operators. Therefore $H_{\bold{z},\bold{s}}$ is  irreducible. 

By Perron-Frobenius theorem 1) $H_{\bold{z},\bold{s}}$ has a unique ground state in the corresponding sector and 2) the  amplitudes of the ground state in each sector are positive.
Since  $X$-type stabilizers  commute with $H_{\bold{z},\bold{s}}$, we conclude that the unique  ground state of $H_{\bold{z},\bold{s}}$ is also an eigenvector of all $X$-type stabilizers. It is the eigenvector corresponding to eigenvalue $+1$  (and not the eigenvalue $-1$), as  the  amplitudes of the ground state are  positive, and an $X$-type operator can only shuffle the amplitudes but cannot change their sign.

The first statement can be proved  as follows. Since $H_p$ can be written as a direct sum of $H_{\bold{z},\bold{s}}$ terms (Eq. \ref{eq:Hpdirectsum}), any ground state of $H_p$ can be written as a linear sum of ground states of $H_{\bold{z},\bold{s}}$ terms. As 
all the $X$-type stabilizers stabilize the unique ground state of each $H_{\bold{z},\bold{s}}$, they also stabilize any ground state of $H_p$. Replacing the role of $X$ and $Z$ operators, we can  similarly prove that any ground state of $H_p$ is stabilized by the  $Z$-type stabilizers of the code. This proves the first statement.

To prove the second statement, we note that since any ground state of $H_p$ is stabilized by the  $Z$-type stabilizers of the code, we conclude that the ground state of $H$ has to be in $\forall i: \, s_i=+1$ sector ($\bold{s}=\overrightarrow{+1}$). 
We already established that ground state of each $H_{\bold{z},\bold{s}}$, and in particular the sectors corresponding to $\bold{s}=\overrightarrow{+1}$, is unique. Therefore we conclude that the ground state of $H=\bigoplus_{\bold{s}}   H_{\bold{z},\bold{s}}$ is unique. 

The third statement in the lemma is a direct consequence of the fact that $\Pi_0$ can be written as $\bigoplus_{\bold{z}} \ketbra{\tilde{\psi}_0}$ (the second statement) and the fact that gauge operators commute with all the logical operators, hence can be written as $g_i=\bigoplus_{\bold{z}}    \tilde{g_i}$. From this we conclude that: 
\ba
\Pi_0 g_i \Pi_0=\bigoplus_{\bold{z}}    \ketbra{\tilde{\psi}_0}\tilde{g_i}\ketbra{\tilde{\psi}_0}=\braket{\tilde{\psi}_0|\tilde{g_i}|\tilde{\psi}_0}\bigoplus_{\bold{z}} \ketbra{\tilde{\psi}_0} = \braket{\tilde{\psi}_0|\tilde{g_i}|\tilde{\psi}_0} \Pi_0,
\ea
where we have used the crucial fact that both $\ket{\tilde{\psi}_0}$ and $\tilde{g_i}$ are independent of $\bold{z}$.
\end{proof}

\subsection{Remarks}

By definition, bare-logical operators of any subsystem code commute with all the elements of its gauge group. This commutation relation induces a tensor structure on the full Hilbert space $\mathcal{H}  \cong \mathcal{H}_{logical} \otimes \mathcal{H}_{rest}$ \cite{ZLL:04}.
Checking this property is easy for our introduced quantum code. In the left hand side of Fig.~\ref{Fig_code}  bare-logical operators are presented in the canonical form: $2k$  pairs of Pauli operators that anti-commute if they have the same index, and commute otherwise. Hence the logical operators generate the full algebra on the $2k$ logical qubits. Therefore any gauge operator,  e.g. operators in the right hand side of Fig.~\ref{Fig_code},  can be written as $g_{ij}=U (I^{\otimes 2k} \otimes \tilde{g}_{ij})U^\dagger$ for some $U$.

It follows that the penalty Hamiltonian or generally any operator that is only composed of gauge operators,
can be written as $H_p=U (I^{\otimes 2k} \otimes \tilde{H}_p) U^\dagger$ and 
%Using the uniqueness of the eigenprojectors, 
the projector to the ground subspace of $H_p$ can be written as $\Pi_0=U (I^{\otimes 2k} \otimes \tilde{\Pi}_0) U^\dagger$. In general there is no constraint on the rank of $\tilde{\Pi}_0$. As was shown in Lemma~\ref{lemma1gaugeHamiltonian}, if $H_p$ includes of a complete set of generators of the gauge group then $\tilde{\Pi}_0$ has rank one. Hence,
\ba
\Pi_0 g_{ij}\Pi_0 &=& U(I \otimes \ket{\psi_0}\braket{\psi_0|\hat{g}_{ij}|\psi_0} \bra{\psi_0})U^\dagger\\
&=& \braket{\psi_0|\hat{g}_{ij}|\psi_0} \Pi_0,
\ea
which is the third statement of the Lemma~\ref{lemma1gaugeHamiltonian}.

%
 
%Note that from this we conclude that for any one-local term $\sigma$ it holds that $\Pi_0 \sigma   \Pi_0=0$. This is proved by observing that $\Pi_0 Z_i   \Pi_0=\Pi_0 Z_i  S_x \Pi_0=-\Pi_0 S_xZ_i   \Pi_0=-\Pi_0 Z_i   \Pi_0=0$, and then repeating the same argument for $X_i$.

The re-scaling ratio $\braket{\psi_0|\hat{g}_{ij}|\psi_0}$ is upper-bounded 
by $1$ as $\hat{g}_{ij}$ is a Pauli operator with eigenvalues $\pm1$. %Although currently we don't have a lower bound on the value of $\alpha$,  
We do not expect these values to be vanishingly small because of the following observation: 

By Perron-Frobenius theorem, the ground state $\ket{\psi}$ has only positive amplitudes. Consider an $XX$ gauge operator. $XX\ket{\psi}$ also only has positive amplitudes. Therefore $\braket{\psi|XX|\psi}$ is a sum of positive numbers, and there is no cancelation by negative numbers. ($H_p$ is symmetric with respect to switching the role of $XX$ and $ZZ$ and therefore the same observation holds for   $\braket{\psi|ZZ|\psi}$.) Indeed for a full quantum compass model, the value of $\braket{XX}$ is around 0.5~\cite{brzezicki2013symmetry}.  (See Figure 5 of Ref.~\cite{brzezicki2013symmetry}, where for the symmetric compass model %, corresponding to alpha=0.5 in the notation of  Eq. 2.1, 
$\braket{XX}$ is around 0.5.)

Note that the exact value of the rescaling ratio depends on the architecture of $H_p$, for example, whether the architecture is 1D or 2D. But for a particular architecture capable of performing universal quantum computation one might be able to analytically evaluate the rescaling ratio by  designing the penalty Hamiltonian. Even if the analytical calculation is not feasible for a particular architecture, for any penalty Hamiltonian these rescaling ratios can be easily measured (and stored) as part of a one-time calibration process (by setting the computational Hamiltonian to zero and measuring the energies of the gauge operators).

%A fast decreasing $\alpha$ would be  daunting. %, and I doubt if it can be found easily for general systems.
%Although currently, we don't have any lower bound on the value of $\alpha$ but the following observation  

\section{2D Layout}

The quantum code described in the main text is capable of encoding any 2D Hamiltonian. Depicting  $XX$ interactions by blue lines and $ZZ$ interactions by orange lines the penalty Hamiltonian and logical operators are presented in Figure \ref{fig:2D}.
%
%\begin{figure}[!t]
%\begin{center}
%\includegraphics[height=0.40\textheight]{Fig_subsystemcode.pdf}
%\caption{\textbf{Description of the $[[6k,2k,2]]$ subsystem code} }
%\label{Fig_continuoussolutions}
%\end{center}
%\end{figure}
\begin{figure}[!ht]
\centering
\begin{subfigure}{.5\textwidth}
  \centering
  \includegraphics[width=1\linewidth]{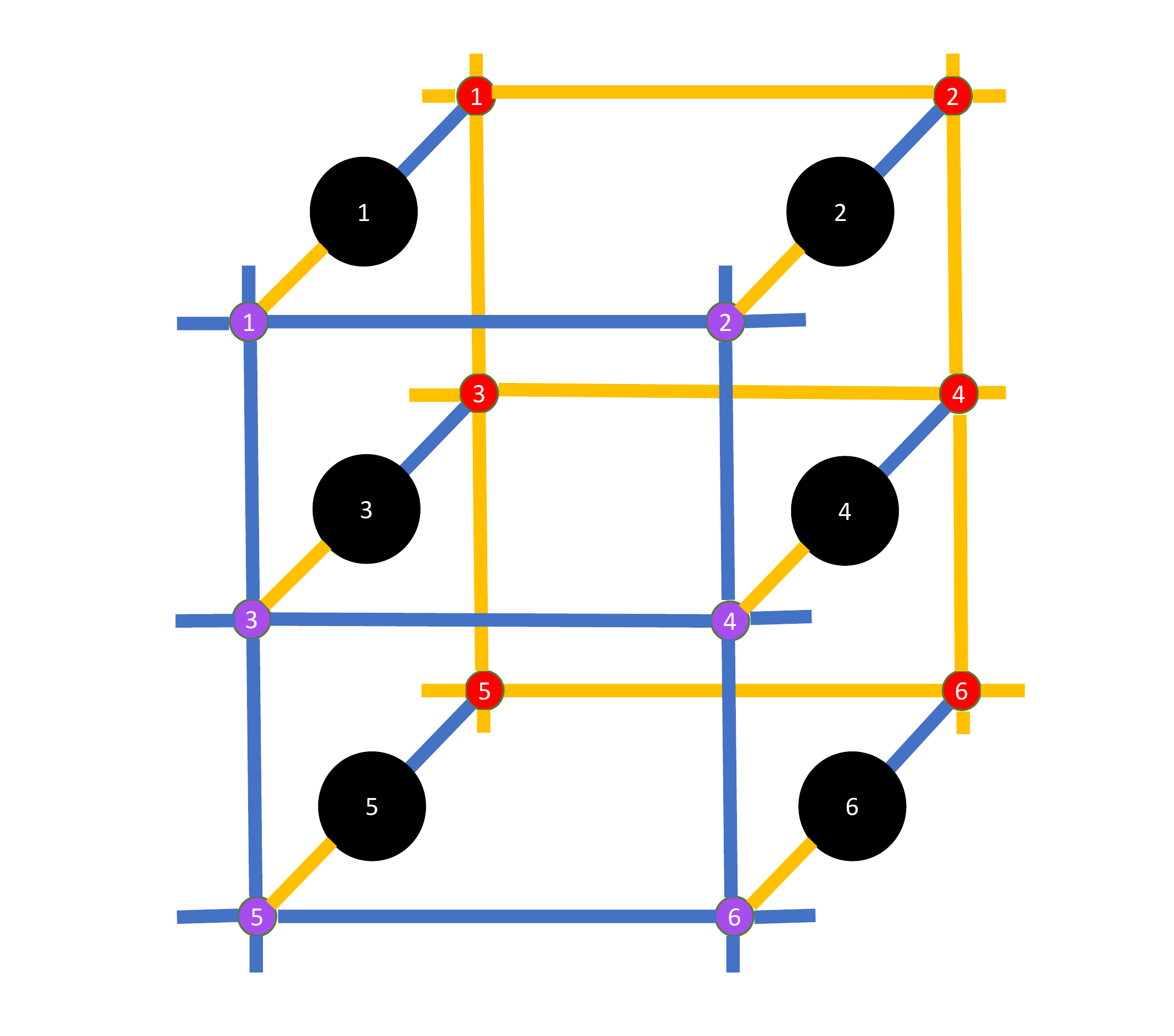}
  \caption{The penalty Hamiltonian}
  \label{fig:sub1}
\end{subfigure}%
\begin{subfigure}{.5\textwidth}
  \centering
  \includegraphics[width=1\linewidth]{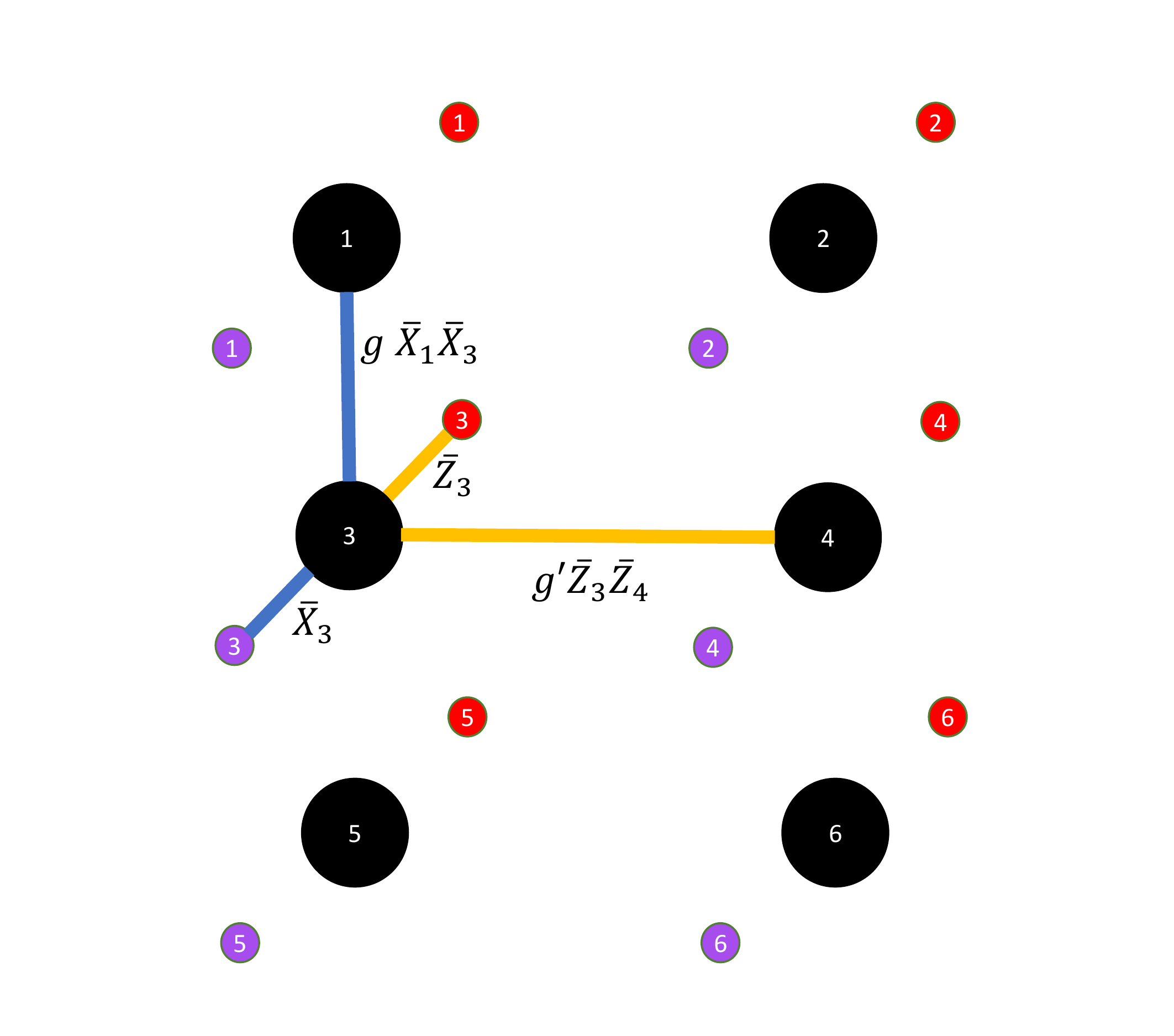}
  \caption{Logical operators}
  \label{fig:sub2}
\end{subfigure}
\caption{Configuration of the code to encode  Hamiltonians defined on a 2D lattice}
\label{fig:2D}
\end{figure}

\section{The proposed code in the form of generalized Bacon-Shor code of Ref.~\cite{PhysRevA.83.012320}} \label{app:Amatrix}
As introduced in Ref.~\cite{PhysRevA.83.012320}, any square binary matrix $A$ defines a  generalized Bacon-Shor code.
Each non-zero elements of $A$ represents a physical qubit. The gauge group $\mathcal{G}$ is generated by  $X_a X_b$ operator for each pair of qubits $a,b$ in the same row and $Z_a Z_b$ operators for each pair of qubits $a,b$ in the same column.

All the properties of this code can be related to properties of the $A$ matrix. The number of logical qubits encoded is equal to the rank of $A$ over ${\mathbb{F}}_2$ and the code distance is the minimum weight of the non-zero vectors in the row space and the non-zero vectors in the column space. So we have $[[n,k,d]]=[[|A|, \mathrm{rank}(A), \min\{d_{\mathrm{row}}, d_{\mathrm{col}}\}]]$, where $|A|$ denotes the Hamming weight of matrix $A$.

%Given the gauge group, in principle all properties of the code can be derived using: 
%\bes
%\begin{align}
%\mathcal{S}&=\mathcal{G} \cap C(\mathcal{G})\quad (\textrm{stabilizer group})\ ,\\
%\mathcal{L}&= C(\mathcal{S})/\mathcal{G} \quad (\textrm{dressed logical operators})\ ,\\
%d&=\min_{P \in C(\mathcal{S})/\mathcal{G} }{|P|} \quad (\textrm{code distance})\ .
%\end{align}
%\ees
%
%where $C$ is the centralizer of the group.

In this notation, our proposed code can be represented using the following $A$ matrix:
\ba
A=
\begin{bmatrix}
  1 &1 &0 & 0 \dots &0 &0 \\
  1 &0 &1 & 0\dots&0 &0\\
  1 &0 &0 & 1\dots&0 &0\\
 \vdots & \vdots\\
 1 &0 &0 & 0...&1 &0\\
   1 &0 &0 & 0...&0 &1\\
  0 &1 &1 & 1...&1 &1
\end{bmatrix}_{(2k+1) \times (2k+1)}.
\ea

To relate this matrix to the arrangement of the qubits in Figure~\ref{Fig_code}, we can associate the diagonal non-zero elements in the matrix with the black qubits in Figure~\ref{Fig_code}, the elements in the first column with the red qubits and the elements in the last row with the lavender qubits. 

It is straightforward to check that $[[n,k,d]]=[[|A|, \mathrm{rank}(A), \min\{d_{\mathrm{row}}, d_{\mathrm{col}}\}]]=[[6k,2k,2]]$.

\subsection{Gauge operators in the canonical form} \label{secappendixGauge}
The code uses $6k$ physical qubits and two stabilizers. Therefore, the dimension of the codespace is $2^{6k}/2^2= 2^{6k-2}$. The logical interactions $\bar{X}_i=X_{B_i} X_{L_i}$ and $\bar{Z}_i=Z_{B_i} Z_{R_i}$ (the left-hand side of Figure \ref{Fig_code}) define $2k$ qubits in the information subsystem. A canonical representation of the remaining $4k-2$ qubits in the gauge subsystem can be defined using the following pairs of operators:
\ba
\forall 1\leq i \leq 2k-1&:& \quad X_{L_i} X_{L_{i+1}}, \prod_{j=1}^{i} Z_{L_j} Z_{B_{j}}\\
\forall 1\leq i \leq 2k-1&:& \quad  Z_{R_i} Z_{R_{i+1}},X^{i}_{L_{2k}} \prod_{j=i+1}^{2k} X_{L_j} X_{B_j}X_{R_{j}},
\ea
%(Choosing the last set of operators symmetrically causes problem as they won't commute with $\prod_{j=1}^{i} Z_{L_j} Z_{B_{j}}$ for the pairs with odd number of overlapping qubits.)

These %$2(4k-2)$ 
$(8k-4)$  operators plus the two stabilizers generate the gauge group. Another set of $8k-2$ 
% $2*2k+ 2*(2k-1)=8k-2$
generators, consists of  the two-body operators $\{X_{B_i} X_{R_i}, X_{L_i} X_{L_{i+1}}, Z_{B_i} Z_{L_i}, Z_{R_i} Z_{R_{i+1}}\}$,  depicted in the right-hand side of Figure \ref{Fig_code}.

\section{Example: Penalty Hamiltonian for the $[[6,2,2]]$ code}

\begin{figure}[!t]
\begin{center}
\includegraphics[height=0.25\textheight]{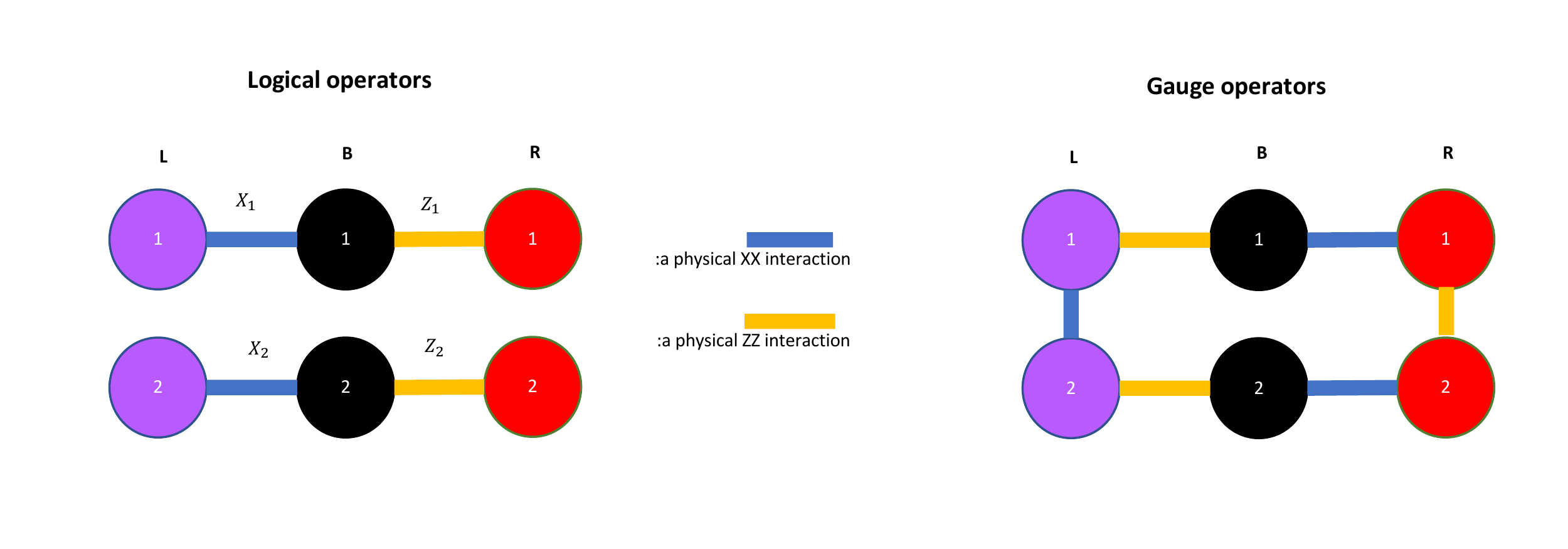}
\caption{\textbf{Description of the $[[6,2,2]]$ subsystem code} }
\label{Fig_code2}
\end{center}
\end{figure}

\ba
H_p= - (X_{L_1}X_{L_2} + X_{B_1}X_{R_1} + X_{B_2}X_{R_2} +Z_{L_1}Z_{B_1} +Z_{L_2}Z_{B_2} +Z_{R_1}Z_{R_2})
\ea

First, we define a new set of operators as follows:
\ba
(\bar{X}_1=)X_{L_1}X_{B_1}  \rightarrow \hat{X}_1\\
(\bar{Z}_1=) Z_{B_1}Z_{R_1} \rightarrow \hat{Z}_1\\
(\bar{X}_2=) X_{L_2}X_{B_2}  \rightarrow \hat{X}_2\\
(\bar{Z}_2=) Z_{B_2}Z_{R_2} \rightarrow \hat{Z}_2
\ea
which correspond to logical qubits, and 
\ba
X_{L_1}X_{L_2} \rightarrow \hat{X}_3\\
Z_{L_1}Z_{B_1} \rightarrow \hat{Z}_3\\
X_{B_2}X_{R_2} \rightarrow \hat{X}_4\\
Z_{R_1}Z_{R_2} \rightarrow \hat{Z}_4\\
\ea
which correspond to  gauge qubits, and
\ba
X_{L_1}X_{L_2} X_{B_1}X_{R_1} X_{B_2}X_{R_2}  \rightarrow \hat{X}_5\\
Z_{R_1} \rightarrow \hat{Z}_5\\
X_{L_2}\rightarrow \hat{X}_6\\
Z_{L_1}Z_{L_2} Z_{B_1}Z_{R_1} Z_{B_2}Z_{R_2} \rightarrow \hat{Z}_6
\ea
corresponding to the stabilizers and error generators. (See the previous section for the mapping for any $k$).

These $6$ pairs of operators with a hat are in canonical form and generate the full algebra for a $6$ qubit Hilbert space. This is basically the recipe for $U$ in $H_p=U (I^{\otimes 2k} \otimes \tilde{H}_p) U^\dagger$.

We can rewrite the penalty Hamiltonian using these operators as:
\ba
H_p=I \otimes I \otimes (- \hat{X}_3- \hat{X}_5  \hat{X}_3  \hat{X}_4- \hat{X}_4- \hat{Z}_3-\hat{Z}_6  \hat{Z}_3  \hat{Z}_4 - \hat{Z}_4)
\ea

One can check that indeed the ground state of $\tilde{H}_p$ is unique, and $\alpha=0.6667$ in this case.

\section{Limitations on the the algebra generated by two-local bare-logical operations in Bravyi's subsystem codes}
Here we prove that the bare-logical operations that can be performed using weight-two operators in Bravyi's generalization of Bacon-Shor subsystem codes Ref.~\cite{PhysRevA.83.012320} are limited.  In particular, we show that they require bare-logical operators with weight larger than two to be able to generate logical operations equivalent to a simple set of interactions such as $\{Z_1, Z_1Z_2, Z_1Z_3, X_1\}$ or $\{Z_1Z_2, Z_1Z_3, Z_1 Z_4,  X_1\}$. Therefore, using Bravyi's subsystem code and only using two-local interactions (with bare-logical encoding) one cannot go beyond the encoding of an Ising \textit{chain} in a transverse field as constructed in Ref.~\cite{Marvian-Lidar:16}.

The proof goes as follows:

\begin{enumerate}
\item First a couple of reminders from Bravyi's original paper~\cite{PhysRevA.83.012320}: 
\begin{enumerate}
\item  X-type logical operators and X-type stabilizers are a combination of columns of matrix A (Eq. 8 of Ref.~\cite{PhysRevA.83.012320})

\item Z-type logical operators and Z-type stabilizers are a combination of rows of matrix A, (Eq. 9 of Ref.~\cite{PhysRevA.83.012320})

\item The distance of the code is the minimum weight of columns and rows in the space of columns and rows (Theorem 2 in Ref.~\cite{PhysRevA.83.012320}). Therefore to have $d>1$,  each column and row of A have at least two non-zero elements.
\end{enumerate}

\item We require that the bare-logical operator $Z_1$ and also bare-logical operators $Z_1Z_2$, $Z_1 Z_3$ all can be implemented using two-body interactions (each up to a $Z$-type stabilizer freedom). As the weight of each rows of A is at least two, each of these three interactions has to be a distinct row of A. (Note that an interaction described by two rows of A requires at least 4-body interaction and so on.)

\item Similarly a bare-logical $X_1$ (up to a $X$-type stabilizer) has to be presentable with a single columns of A. But such an interaction has to anticommute with $Z_1, Z_1Z_2, Z_1Z_3$. As these three interactions are distinct rows of A, this means that a bare logical $X_1$ has to have at least three non-zero elements in the corresponding rows! Therefore $X_1$ cannot be two-local.
\end{enumerate}

Similar argument works for  $\{Z_1Z_2, Z_1Z_3, Z_1 Z_4,  X_1\}$. 

\subsection{No-go theorem: General CSS codes}
\begin{theorem}
Consider any nontrivial ($d>1$) stabilizer subsystem code of CSS type with a gauge group that can be generated using  $XX$ and $ZZ$ interactions. Then the weight of an X-type (Z-type)  single-qubit bare-logical operator is lower-bounded by the number of Z-type (X-type) bare-logical operators acting on its supporting logical qubits.

\end{theorem}

\begin{proof}
First, note that for any code that can detect any single error, each physical qubit has to be acted upon with at least two gauge generators with different types. Therefore for each physical qubit, there are at least two gauge generators that their action  anti-commute on this qubit.

Let $\bar{L}_1$ and $\bar{L}_2$ be two distinct two-local bare-logical operators of the same type. We show that the support of $\bar{L}_1$ and $\bar{L}_2$ cannot overlap, i.e., they act on a distinct set of physical qubits.

To see this, note that, clearly, any two distinct two-local operators of the same type can at most overlap on one physical qubit (otherwise, they are not distinct!) Assume that they overlap on only one physical qubit. By definition, bare-logical operators have to commute with all the gauge operators. But this is not possible since, as discussed earlier, there always exists a weight-two gauge operator that anti-commutes with either $\bar{L}_1$ or $\bar{L}_2$ (there always exists a gauge operator whose action on the overlapping qubit is different from the type of $\bar{L}_1$ and $\bar{L}_2$.)  This leads to a contradiction and we conclude that the supports of the two bare-logical operators of the same type cannot overlap.

Suppose $\bar{L}_1$ to $\bar{L}_m$ are distinct two-local bare-logical operators of the same type whose actions on logical qubit $a$ commutes. Then implementing the single logical operator that anti-commuted with the action of all these operators on the logical qubit $a$ requires $m$-body interaction. This trivially follows from the fact that $\bar{L}_1$ and $\bar{L}_2$ are supported on a non-overlapping set of qubits and the support of any operator that anti-commutes with all of them has to overlap with their support.

\end{proof}

\end{document}